\documentclass[lettersize,journal]{IEEEtran}
\usepackage{amsmath,amsfonts}
\usepackage{mathtools}
\usepackage{algorithm}
\usepackage{algorithmic}
\usepackage{array}
\usepackage[caption=false,font=normalsize,labelfont=sf,textfont=sf]{subfig}
\usepackage{textcomp}
\usepackage{stfloats}
\usepackage{url}
\usepackage{verbatim}
\usepackage{graphicx}
\usepackage{cite}
\usepackage{dsfont}
\usepackage{amssymb}
\usepackage{color}
\usepackage[colorlinks=true, linkcolor=blue, citecolor=blue, urlcolor=blue]{hyperref}
\usepackage{tabularx}
\usepackage{tikz}
\newcommand{\tausig}{\tau_t^2 + \sigma_h^2}

\newcommand{\mbf}[1]{\mathbf{#1}}
\DeclareMathOperator*{\argmax}{argmax}
\usepackage{amsthm}
\newtheorem{theorem}{Theorem}

\newtheorem{lemma}{Lemma}

\begin{document}

\title{Sparse Regression Codes for Non-coherent \\ SIMO channels}

\author{
\IEEEauthorblockN{Sai Dinesh Kancharana, Madhusudan Kumar Sinha, Arun Pachai Kannu.}
      \thanks{This work is supported in part by the Qualcomm 6G UR grant and partly by the Prime Minister's Research Fund (PMRF).
      The authors are with the Department of Electrical Engineering, Indian Institute of Technology Madras, Chennai - 600036, Tamil Nadu, India. (email: ee20d401@smail.iitm.ac.in, ee16d028@smail.iitm.ac.in, arunpachai@ee.iitm.ac.in).
      }
}




\maketitle

\begin{abstract}
Motivated by hyper-reliable low-latency communication in 6G, we consider error control coding for short block lengths in multi-antenna fading channels. In general, the channel fading coefficients are unknown at both the transmitter and receiver, which is referred to as non-coherent channels. Conventionally, pilot symbols are transmitted to facilitate channel estimation, causing power and bandwidth overhead. Our paper considers sparse regression codes (SPARCs) for non-coherent flat-fading channels without using pilots. We develop a novel greedy decoder for SPARC using maximum likelihood principles, referred to as maximum likelihood matching pursuit (MLMP). MLMP works based on successive combining principles as opposed to conventional greedy algorithms, which are based on successive cancellation. We also obtain the noiseless perfect recovery condition for our successive combining algorithm. In addition, we develop an approximate message passing (AMP) SPARC decoder for the non-coherent flat fading model. Using simulation studies, we show that the MLMP decoder for SPARC outperforms AMP and other greedy decoders. Also, SPARC with MLMP decoder outperforms polar codes employing pilot-based channel estimation and polar codes with non-coherent decoders.
\end{abstract}

\begin{IEEEkeywords}
Sparse regression codes, short block lengths, non-coherent detection, greedy algorithm, approximate message passing
\end{IEEEkeywords}

\section{Introduction}
\IEEEPARstart{T}{he}
 rapid proliferation of smart devices and the advancement of wireless systems underscore the importance of low-latency communications (LLC). A key use case in next-generation 6G networks is hyper-reliable low-latency communication (HRLLC), which targets ultra-low latency (on the order of 0.1–1 ms) and extremely high reliability to support applications such as industrial automation, telemedicine, and autonomous driving \cite{HRLLC1,HRLLC2arxiv,HRLLC_transport}. Conventional error control codes used in modern cellular communications use very large block lengths. However, for scenarios such as the Internet of Things (IoT), where many devices are simultaneously connected or mission-critical applications such as Vehicle-to-everything (V2X) communications, where the latency needs to be very low, short block length codes are suitable.  
 
 In typical wireless fading environments, the channel coefficients are unknown a priori. Hence, pilot symbols are transmitted apart from the data, in order to facilitate channel estimation followed by data detection \cite{apk1}. These pilot symbols consume power and bandwidth overhead, reducing the resources available for data transmission, which can have a significant impact, especially for short block length codes.  Channel estimation errors will further affect the data detection performance. 
 In this paper, we consider sparse regression codes (SPARCs)  with short block lengths for unknown flat-fading single-input multiple-output (SIMO) channels. We develop non-coherent decoders for SPARC that do not require knowledge of the channel fading coefficients.  

\subsection{Existing literature on SPARCs}
Sparse regression codes (SPARCs), also called Sparse superposition codes, were developed by Barron and Joseph \cite{barron_joseph_2011_ISIT,joseph_barron_2012_TIT,joseph_barron_2014_TIT} for error correction over the AWGN channel. SPARC codewords are encoded based on a measurement matrix $\mbf{A}$ of size $N \times L$, partitioned into $K$ equal-sized sections. Based on the information bits, one column is chosen from each section, and the final codeword is formed as a linear combination of the chosen columns. Hence, the codeword can be represented as $\mbf{s} = \mbf{Ax}$, where $\mbf{x}$ is a $K$-sparse $L$-length vector. Hence, decoding from the SPARC codeword over a noisy channel is essentially a sparse recovery problem. The nonzero values of $\mbf{x}$ can take values from an $M$-ary constellation to encode additional bits into the SPARC codeword. The code rate of the SPARC is given by 
\begin{equation}
    R = \frac{K}{N}\log_2 \frac{L}{K} + \frac{K}{N}\log_2 M.
\end{equation}
The authors in \cite{barron_joseph_2011_ISIT,joseph_barron_2012_TIT,joseph_barron_2014_TIT} proposed an \emph{adaptive successive decoder} for which the probability of decoding error decays exponentially to zero with increasing block length while having poor empirical performance at finite block lengths. Later, \cite{barron_cho_2012_ISIT} improved the empirical performance while giving the same asymptotic guarantees. Using a random Gaussian dictionary matrix and an approximate message passing algorithm (AMP), the authors in \cite{2017_Crush_unmod_sparc} showed that SPARCs in the AWGN channel can achieve arbitrarily low section error rates for all rates $R$ less than channel capacity $C$, i.e., $R < C$. Using an alternate formulation, the authors in \cite{2017_barbier_AMP_SPARCs} also demonstrated the capacity-achieving nature of SPARCs with AMP using replica analysis and spatially coupled matrices. In addition, techniques such as spatial coupling \cite{2015_barbier_approximate,2017_barbier_AMP_SPARCs,2019_SPARC_GAMP} and optimal power allocation \cite{2017_greig_Techniques} were shown to improve empirical performance at finite block lengths. In \cite{2021_Crush_spatialcoupling} and \cite{2021_Ramji_mod_sparc}, the authors proved that SPARCs with spatially coupled matrices and modulated SPARCs can also achieve arbitrarily low section errors with sufficiently large block lengths and rates less than capacity.  Hadamard-based matrices \cite{2017_Crush_unmod_sparc,2017_barbier_AMP_SPARCs} and Fast Fourier Transform (FFT) matrices \cite{2021_Ramji_mod_sparc} have also been shown to work well with the AMP decoder while having lower computational complexity. Orthogonal AMP \cite{2017_oamp} and vector AMP \cite{2019_rangan_vamp} were introduced for more general dictionary matrices that were non-Gaussian. Recently, the authors in \cite{SC-VAMP_SSC} introduced spatially coupled VAMP (SC-VAMP) and proved that it is asymptotically capacity-achieving, and in \cite{Binomial_SPARCs}, SPARCs were shown to achieve capacity asymptotically with binomially distributed design matrices with maximum likelihood decoding. In \cite{err_prob_SC_SPARC}, the authors rigorously analyse the error probability of SPARCs on the AWGN channel in the non-asymptotic regime, and prove an exponentially decaying error probability with increasing block length. For low to medium code rates, clipped SPARCS were introduced in \cite{liang2021finite,irregularly_clipped_SPARCs} to improve their block error rate at finite lengths, and Block-orthogonal Sparse regression codes were introduced for very low code rates \cite{BOSS_2021,BOSS_2022,BOSS_2025}. 

Although the above papers suggest that large block length SPARCs are capacity-achieving, short block length codes are better suited for many low-latency applications. In modern communication systems, such as 5G networks, control channels usually employ short block length codes. Applications like IoT and massive machine-type communications (mMTC) require low latency, where short block length codes can prove to be helpful. Toward this end, finite block length analysis has gained significant attention \cite{2010_yury_finiteblocklength}. The reader is referred to the works in \cite{cocskun2019efficient, van2018short,shortblock_URLLC} for a study on the short block performance of existing codes. Deterministic sequences with very low correlation properties exist in the CDMA literature \cite{CDMA_dict1,CDMA_dict2} and in quantum information theory \cite{woottersMUB,MUB2}. Dictionary matrices formed from these sequences, specifically gold codes \cite{CDMA_dict1} and MUB sequences from \cite{woottersMUB}, have columns $L \approx N^2$, where $N$ is the length of the code. The maximum normalised correlation between different columns of the matrix, also known as mutual coherence, is approximately $1/\sqrt{N}$. When we choose the sparsity level $K$ of the sparse vector of the order of $\sqrt{N}$, the code rate falls in the regime of interest for very short block lengths. 

SPARCs for the short block length regime over the AWGN channel have been studied and analysed in \cite{2024_madhu_MAD}, where the authors have proposed the \emph{match and decode} (MAD) algorithm, based on maximum likelihood decoding. In the non-coherent SIMO setting, the sparse vector code (SVC) \cite{SVCji2018sparse}, which is closely related to SPARC, has been studied in \cite{ji2019pilot, PLOS_DL,SVC_IIOT}. In \cite{SVCji2018sparse,MMP_SSC}, the authors used the \emph{multi-path matching pursuit} algorithm, which extends the greedy \emph{orthogonal matching pursuit} (OMP) algorithm \cite{OMP, eldar2010block}. In \cite{PLOS_DL}, pilotless one-shot transmission and a deep neural network-based decoder were proposed. SVC had been applied for MIMO channels in \cite{SVC_MIMO}. Later, SVC had also been extended to encode additional bits using constellation symbols in \cite{ModulatedSVC}. 

\subsection{Our contribution}
In this work, we develop a SPARC decoder for the non-coherent flat fading channel with multiple receive antennas. The optimal maximum likelihood (ML) decoder for SPARCs is derived but is computationally infeasible. Inspired by the non-coherent ML metric, we develop a novel greedy SPARC decoder, which we refer to as the maximum likelihood matching pursuit (MLMP) algorithm. Similar to \cite{MMP_TIT2014}, we also introduce a parallel search mechanism for MLMP to mitigate the propagation of errors in the greedy algorithm. 
Due to the presence of multiple receive antennas, the information-bearing sparse signal in SPARC has a joint-sparse structure. 
Existing block sparse recovery algorithms, such as block orthogonal matching pursuit \cite{ji2019pilot} and block approximate message passing \cite{2021_Schniter_LAMP,2013_Donoho_BlockAMP}, can be employed to decode the SPARC.
The contributions of our work are summarised as follows:
\begin{itemize}
    \item We develop a non-coherent maximum likelihood (ML) detector for SPARCs in SIMO flat-fading channels. Since the search space of the ML decoder is prohibitively large, we develop a novel greedy algorithm for SPARC decoding using partial ML metrics.  While conventional greedy sparse signal recovery algorithms are based on successive cancellation techniques, our greedy maximum likelihood matching pursuit (MLMP) decoder is based on the successive \emph{combining} principle. We improve the MLMP algorithm by introducing a parallel search mechanism, which we call parallel-MLMP (P-MLMP).

    \item We develop guarantees on the decoding success of the MLMP algorithm based on the number of sections $K$ and the mutual coherence of the directional matrix. Interestingly, our guarantee coincides with that of the well-known OMP algorithm \cite{OMP}.
    

    \item We derive the approximate message passing decoder for SPARCs (SAMP) in a non-coherent flat-fading SIMO channel. We obtain the minimum mean squared error (MMSE) denoiser for SAMP, exploiting the SPARC structure that exactly one nonzero entry is present in each section. We show that the mean squared error (MSE) of SAMP at each iteration is accurately tracked by a scalar recursion known as \emph{state evolution}.

    \item We have shown through simulation that the proposed MLMP and P-MLMP significantly outperform the state-of-the-art sparse recovery algorithms such as Block-OMP \cite{OMP,eldar2010block} and its variant known as Modified Block-OMP, designed to take advantage of the structure of SPARC. We also show that parallelisation can further improve the BLER performance of these greedy algorithms, and more significantly at higher code rates.


    \item We also show that SPARC with the MLMP decoder outperforms existing error control codes such as Polar codes with pilot-aided transmission (PAT-Polar) and a non-coherent version of polar codes \cite{yuan2021polar}, in terms of block error rate (BLER) performance. 

\end{itemize}


\subsection{Outline of the paper}
The rest of the paper is organised as follows. In Section \ref{Background}, we describe the encoding of SPARC for the non-coherent SIMO flat-fading model. In Section \ref{MLMP section}, we describe the MLMP algorithm as an approximation of the ML detector. We also introduce Parallel-MLMP (P-MLMP) and Modified Block-OMP (MBOMP). In Section \ref{AMP section}, we describe the SPARC approximate message-passing algorithm (SAMP) under the current system model. In the next section \ref{Polar codes section}, we describe pilot-based and non-coherent polar codes. We also describe the sphere packing bound for the coherent SIMO channel and compare the computational complexities of the decoding algorithms discussed in this paper. Section \ref{Results section} shows the performance of MLMP and SAMP through simulations. We conclude with some closing remarks and future directions in section \ref{conclusion}, followed by a proof of Theorem \ref{Theorem1} in the appendix \ref{proof_thm1}.

\subsection{Notations}
Scalars are small case letters $x$, vectors, and matrices are in bold lower case $\mathbf{x}$ and upper case $\mathbf{X}$, respectively. $|x|$ represents the absolute value of $x$. $\mathbf{X}^*$ is the hermitian of $\mathbf{X}$, and $|\mbf{X}|$ denotes it determinant. $\|\mathbf{x}\|$ is the $\ell_2$ norm of $\mbf{x}$. $\langle \mathbf{a},\mathbf{b}\rangle$ is the inner product between $\mathbf{a}$ and $\mathbf{b}$. $\log$ and $\log_2$ represent the natural and base-2 logarithms, respectively. For any positive integer $N$, $[N]$ represents the set \{1, 2, \dots ,N\}.  $|\mathcal{X}|$ is the size of the set $\mathcal{X}$. $\mathcal{CN}(\mu,\sigma^2)$ is the complex Gaussian pdf with mean $\mu$ and variance $\sigma^2$. For a given matrix $\mbf{X}$, $\mbf{X}_{\mathcal{G},:}$ denotes the matrix formed from the rows of $\mbf{X}$ with indices from the set $\mathcal{G}$, and consequently $\mbf{X}_{k,:}$ denotes the row-$k$ of the matrix $\mbf{X}$.

\section{Background/Preliminaries} \label{Background}

\subsection{SPARC encoding} 
SPARCs are defined by a dictionary matrix $\mathbf{A}$ of size $N \times L$ where $L \geq N$. Let $\mathcal{A} = \{\mathbf{a}_1,\mathbf{a}_2, \dots ,\mathbf{a}_L\}$ denote the set of columns of $\mathbf{A}$. This matrix is partitioned into $K$ disjoint sections, with $\mathcal{A}_k$ denoting the set of columns in the $k$\textsuperscript{th} section and $\mathcal{Q}_k$ denoting the set of column indices in $\mathcal{A}_k$. $\sec(j)$ denotes the set of indices of all columns in the section that contains the column $\mbf{a}_j$. For example, if the column $\mbf{a}_j$ belongs to the first section $\mathcal{A}_1$, then $\sec(j) = \mathcal{Q}_1$. Sections can be of unequal sizes, and an optimal section partitioning algorithm from \cite{2024_madhu_MAD} can be used to maximise the code rate based on $L$ and $K$. The optimal partitioning algorithm outputs section sizes $L_k$ corresponding to section $\mathcal{A}_k$, which are a power of two for any $k$. 
Based on the information bits, the encoder selects one column from each section $\mathcal{A}_k$, $\forall k \in [K]$, and the codeword is formed as a sum of the columns chosen from all sections. Since each column in $\mathcal{A}_k$ can be indexed using $\log_2{L_{k}}$ bits, the number of bits conveyed by the SPARC codeword is $N_b = \sum_{k=1}^{K}\log_{2} L_k$. Let the support set $\mathcal{S} \subset [L]$ contain the indices of the columns chosen from $\mathcal{A}$. Then, we can represent the SPARC codeword as 
\begin{equation}
    \mathbf{s}  = \sum_{m \in \mathcal{S}} \mathbf{a}_{m} = \mathbf{Ax},
\end{equation}
where $\mathbf{x}$ is a $K$-sparse column vector with exactly $K$ nonzero entries (equal to one) at the corresponding indices of the chosen columns. Note that we do not use modulation symbols in the nonzero locations of $\mbf{x}$ since the codeword is transmitted over the fading channel, and this information will be lost. Let $\mathcal{C}$ denote the complete set of all possible codewords that can be formed. Since each section contains $L_{k}$ columns, $\forall k \in [K]$, the total number of codewords in the set $\mathcal{C}$ is $|\mathcal{C}| = \prod_{k=1}^{K}L_k$. If all sections are of equal size, then the total number of codewords is $|\mathcal{C}| = \left ( \frac{L}{K} \right )^K$. The code rate $R$ for SPARC is measured in bits per real channel use (bpcu), where $R = N_b/N$ bpcu for a real dictionary matrix and $R = N_b/(2N)$ bpcu for a complex dictionary matrix.

\subsection{MUB Matrices}
The mutual coherence property of a dictionary matrix defined as
\begin{equation} \label{mutual_coherence}
    \mu(\mathbf{A}) = \max_{ i \neq j} \frac{|\langle \mathbf{a}_i,\mathbf{a}_j \rangle|}{\|\mathbf{a}_i\| \, \|\mathbf{a}_j\|}, 
\end{equation}
plays an important role in sparse signal recovery performance. We use dictionary matrices constructed using mutually unbiased bases (MUBs) \cite{woottersMUB}, with the number of columns $L=N^2$ and $\mu(\mathbf{A}) = \frac{1}{\sqrt{N}}$. Two orthonormal bases $\mathcal{B}_1$ and $\mathcal{B}_2$ in an $N-$dimensional complex inner product space $\mathbb{C}^N$ are mutually unbiased if and only if $|\langle \mathbf{p,q} \rangle | = \frac{1}{\sqrt{N}}$ for any $\mathbf{p} \in \mathcal{B}_1$ and $\mathbf{q} \in \mathcal{B}_2$. When $N$ is a power of a prime number, explicit constructions for MUBs have been given in \cite{woottersMUB}. 


\subsection{SPARC for unknown fading channel}
We consider a non-coherent fading model where both the transmitter and the receiver lack channel state information (CSI). We assume a flat-fading channel, where the fading coefficients remain constant over the entire duration of the codeword and change independently over different codewords. This is a particularly valid assumption, since we are working with very \emph{short block lengths}. This assumption becomes weaker with increasing block lengths. We aim to build computationally efficient non-coherent decoders for SPARCs in SIMO flat-fading channels. With $D$ receive antennas, the observation vector $\mathbf{y}_i \in \mathbb{C}^N$ of size $(N \times 1)$ at each receive antenna is,
\begin{equation} \label{MSPARC_MMV}
    \mathbf{y}_i = h_i \mathbf{s} + \mathbf{v}_i,  \hspace{2mm} \forall i \in [D],
\end{equation}
where $h_i \sim \mathcal{CN}(0,\sigma_h^2)$ is an unknown complex gaussian fading coefficient, and $\mbf{v}_i \sim \mathcal{CN}(0,\sigma_v^2\mathbb{I}_N)$ is the complex white gaussian noise (CWGN) at the  $i$\textsuperscript{th} receive antenna.
Given that the receiver observes $\mathbf{y}_i, \hspace{1mm} \forall i \in [D]$ and has knowledge of $\mathbf{A}$, the decoder's task is to efficiently detect the locations of the columns of $\mathbf{A}$ that formed the codeword. 

In sparse signal recovery, a significant body of literature \cite{marques2018review,blanchard2014greedy} frequently highlights the computational superiority of greedy algorithms compared to conventional convex optimisation-based $\ell_1$ minimisation techniques. In this interest, we propose a greedy algorithm based on ML detection for the model in (\ref{MSPARC_MMV}) discussed in the next subsection.


\section{Maximum likelihood matching pursuit} \label{MLMP section}

\subsection{Maximum Likelihood Detector} \label{ML detector}
The model for SPARCs through an unknown i.i.d. flat-fading SIMO channel is reiterated here as,
\begin{equation} 
    \begin{split} 
        \mathbf{y}_i &= h_i\mathbf{s} + \mathbf{v}_i, \hspace{2mm} \forall i \in [D]. \nonumber
    \end{split}
\end{equation} 
The ML detector for $\mathbf{s}$ from the observations $\mathbf{y}_i, \forall i \in [D]$ is, 
\begin{align}
    \hat{\mathbf{s}} &= \argmax_{\mathbf{s} \in \mathcal{C}} \hspace{1mm} \log p(\mathbf{y}_1, \dots ,\mathbf{y}_D | \mathbf{s})  = \argmax_{\mathbf{s} \in \mathcal{C}} \hspace{1mm}\sum_{i=1}^{D}  \log(p(\mathbf{y}_i|\mathbf{s})). \label{channel_iid} 
\end{align}
From (\ref{MSPARC_MMV}), the conditional distribution $\mathbf{y}_i|\mathbf{s}$ is 
a zero mean complex gaussian with covariance matrix $\mathbf{C}$, i.e.,
\begin{equation}
    \mathbf{y}_i|\mathbf{s} \sim  \hspace{1mm} \mathcal{CN}(0, \underbrace{\sigma_h^2\mathbf{s}\mathbf{s}^* + \sigma_v^2\mathbb{I}_N}_{\mathbf{C}}), \forall i \in [D]. 
\end{equation}
The determinant of $\mbf{C}$ is a product of its eigenvalues, 
\begin{equation}
    |\mathbf{C}| = \sigma_v^{2(N-1)} (\sigma_h^2 \|\mathbf{s}\|^2 + \sigma_v^2). \label{C_det}
\end{equation}
The inverse of $\mbf{C}$ can be calculated using the  Woodbury identity as follows,
\begin{equation}
    \mathbf{C}^{-1} = \frac{1}{\sigma_v^2} \left( \mathbb{I}_N - \frac{\sigma_h^2}{\sigma_v^2+ \sigma_h^2 \|\mathbf{s}\|^2}\mathbf{s}\mathbf{s}^* \right). \label{C_inv}
\end{equation}
Substituting (\ref{C_det}) and (\ref{C_inv}) back into 
the ML detector in (\ref{channel_iid}) we get,
\begin{align} \label{ML_det}
    \mathbf{\hat{s}} 
    &= \argmax_{\mathbf{s} \in \mathcal{C}} \hspace{1mm} \beta_{\mathbf{s}} \sum_{i=1}^{D}| \langle \mathbf{y}_i , \mathbf{s} \rangle|^2 - D\gamma_{\mathbf{s}},
\end{align}
with 
\begin{equation} \label{beta_gamma}
    \beta_{\mathbf{s}}=\frac{\sigma_h^2/\sigma_v^2}{\sigma_v^2 + \sigma_h^2 \|\mathbf{s}\|^2}, \hspace{1mm} \gamma_{\mathbf{s}}= \log \left ( \frac{\sigma_h^2}{\sigma_v^2} \|\mathbf{s}\|^2 + 1 \right ).
\end{equation}

\subsection{Maximum Likelihood Matching Pursuit}
To find the ML estimate of $\mathbf{s}$, we need to compute the metric in (\ref{ML_det}) for all possible codewords in $\mathcal{C}$. 
For example, the multi-antenna ML detector can be elaborately represented as,
\begin{multline}
    \hat{\mathbf{s}} =\argmax_{m_1,m_2 \dots ,m_K} \beta_{\mathbf{s}}\sum_{i=1}^{D}  |\mathbf{y}_i^*\mathbf{a}_{m_1} + \cdots + \mathbf{y}_i^*\mathbf{a}_{m_K}|^2 - D\gamma_{\mathbf{s}}, \label{MLMP_metric_expanded}
\end{multline}
where, $\mathcal{S} = \{m_1, m_2, \dots ,m_K \}$ is the support set that forms the codeword. Solving (\ref{MLMP_metric_expanded}) is infeasible even for small block lengths, as $|\mathcal{C}|$ is quite large. So, we introduce an intelligent way to approximate the ML detector into an iterative decoder named the maximum likelihood matching pursuit (MLMP) algorithm. MLMP iteratively finds the columns that make up the codeword, one at a time. 

At the receive antenna $i,\forall i \in [D]$, the observed signal can be represented as
\begin{align}  
    \mathbf{y}_i &= h_i\mathbf{s} + \mathbf{v}_i ~=~ h_i\mathbf{a}_{m_1} + h_i\sum_{m \in \mathcal{S} \symbol{92} m_1}\mathbf{a}_{m} + \mathbf{v}_i. \label{MLMP_iter1}
\end{align}

In the first iteration, the aim is to detect only one column (w.l.o.g., we consider it as the first one in (\ref{MLMP_iter1})) that made up the codeword $\mathbf{s}$. While estimating this column, the rest of the $|\mathcal{S} \symbol{92} m_1| = K-1 $ undetected columns act as interference. Hence, we set the  effective noise variance in the first iteration  as,  
\begin{equation} \label{effective_noise}
    \Tilde{\sigma}_{v,1}^2 = \sigma_v^2 + \frac{\sigma_h^2}{N}\| \sum_{m \in \mathcal{S} \symbol{92} m_1} \mathbf{a}_m\|^2,
\end{equation}
where, 
\begin{align}
    \|\sum_{m \in \mathcal{S} \symbol{92} m_1} \mathbf{a}_m\|^2 &= \langle \sum_{m \in \mathcal{S} \symbol{92} m_1} \mathbf{a}_m,\sum_{m \in \mathcal{S} \symbol{92} m_1} \mathbf{a}_m \rangle,  \nonumber \\
    &= \sum_{m \in \mathcal{S} \symbol{92} m_1} \|\mathbf{a}_m\|^2 + \sum_{p} \sum_{q \neq p} \langle \mathbf{a}_p,\mathbf{a}_q \rangle. \label{crossterms}
\end{align}  

The cross terms in (\ref{crossterms}) can be bounded using the mutual coherence of the matrix $\mu(\mathbf{A})$ in (\ref{mutual_coherence}). Hence, the sum in (\ref{crossterms}) lies in the interval $$[ (K-1) - (K-1)(K-2)\mu(\mathbf{A}) , (K-1) + (K-1)(K-2)\mu(\mathbf{A})].$$ So, we approximate the sum by taking the mid-point of the interval, and the effective noise variance in  (\ref{effective_noise}) is taken as 
\begin{equation}
    \Tilde{\sigma}_{v,1}^2 \approx \sigma_v^2 + \frac{\sigma_h^2}{N}(K-1).
\end{equation}
In the first iteration of MLMP, we are trying to find a single active column from the codeword by searching over all $\mathbf{a}_m \in \mathcal{A}$. Since all columns have unit norm, the terms $\beta$ and $\gamma$  in the metric (\ref{ML_det}) remain the same $\forall \mathbf{a}_m \in \mathcal{A}$.  Hence, MLMP selects the first active column as
\begin{equation}
    \mathbf{a}_{{\hat{m}_1}} = \argmax_{\mathbf{a}_{m} \in \mathcal{A}} \sum_{i=1}^D|\langle \mathbf{y}_i,\mathbf{a}_{m} \rangle|^2.
\end{equation}



In iteration $k$, the detected columns from the first $k-1$ iterations are assumed to be correct, and we try to find the $k^{th}$ column that is part of the codeword. Hence, the observed signal can be written as,
\begin{equation}
    \mathbf{y}_i = h_i ( \sum_{\ell=1}^{k-1}{\mathbf{a}}_{\hat{m}_\ell}) + h_i \sum_{m \in \mathcal{S} \setminus 
    \{ \hat{m}_1, \dots ,\hat{m}_{k-1}\}}  \mathbf{a}_m + \mathbf{v}_i, \forall i \in [D].
\end{equation}
Now, the remaining $K-k$ columns act as interference and, hence, the effective noise variance is $\Tilde{\sigma}_{v,k}^2 = \sigma_v^2 + \frac{\sigma_h^2}{N} (K-k)$. Inspired by the metric in (\ref{MLMP_metric_expanded}), we use the successive combining method by including the previously detected columns as part of the metric and proceeding to find the next column that is part of the codeword. Hence, $k$\textsuperscript{th} column that is part of the codeword is chosen as,  
\begin{multline} \label{MLMP_kth_iter}
    {\mathbf{a}}_{\hat{m}_k} = \argmax_{\mathbf{a}_{m} \in \mathcal{A} \setminus \{\mathcal{\hat{A}}_{1} \cup \cdots \cup \mathcal{\hat{A}}_{{k-1}} \}} \beta_{k,m} \sum_{i=1}^{D} |\langle \mathbf{y}_i, \sum_{\ell =1}^{k-1} {\mathbf{a}}_{\hat{m}_\ell} + \mathbf{a}_{m} \rangle|^2  \\ - D\gamma_{k,m},
\end{multline}
where, 
\begin{align} \label{beta_gamma_iterk}
    \beta_{k,m} &= \frac{\sigma_h^2/\Tilde{\sigma}_{v,k}^2}{\Tilde{\sigma}_{v,k}^2 + \sigma_h^2\|\sum_{\ell=1}^{k-1} {\mathbf{a}}_{\hat{m}_\ell} + \mathbf{a}_{m}\|^2}, \hspace{1mm}  \\
    \gamma_{k,m} &= \log \left ( \frac{\sigma_h^2}{\Tilde{\sigma}_{v,k}^2} \|\sum_{\ell=1}^{k-1} {\mathbf{a}}_{\hat{m}_\ell} + \mathbf{a}_{m}\|^2 + 1 \right ). 
\end{align}

The MLMP algorithm stops after K iterations. Since the active columns of the codeword are obtained greedily, the search space of MLMP is linear in $K$, while that of the optimal ML detector grows exponentially with $K$.  


\subsection{Parallel-MLMP}
In the MLMP algorithm, detection in the first iteration is the most susceptible to errors, as the remaining $K-1$ undetected columns act as interference. Given the iterative nature of greedy algorithms like MLMP, any error occurring in the first iteration tends to propagate through subsequent iterations, leading to errors in detecting the subsequent columns \cite{MMP_TIT2014}.

To overcome this, we introduce parallel-MLMP (P-MLMP), which chooses the top `$P$' columns with the highest ML metrics in the first iteration. We then run the MLMP algorithm $P$ times in parallel, assuming one of the top $P$ columns from the first iteration as an active column of the codeword. We get $P$ different codewords from these $P$ parallel MLMP decoders.  Among these $P$ candidates, we select the candidate with the largest ML metric as the final chosen codeword. Algorithm \ref{alg:parallel-mlmp} shows the working for Parallel-MLMP.

\begin{algorithm}[h]
\caption{Parallel-MLMP} \label{alg:parallel-mlmp} 
\begin{algorithmic}[1]


\STATE \textbf{Input}: $\mathbf{y}, \mathbf{A}, K, P$.


    \STATE \textbf{First metric}:
    $q(m)= \sum_{i \in [D]}|\langle \mathbf{y}_i,\mathbf{a}_m \rangle|^2$, $\forall \mathbf{a}_m \in \mathcal{A}$.

    \STATE \textbf{Let $\mathcal{J} = (j_1,...,j_P)$ denote the set of indices corresponding to the top $P$ values of $q(m)$ such that $q(j_1)\geq q(j_2) \geq \cdots q(j_P) \geq q(m)$ for any $m \in [L]\setminus\mathcal{J}$}.
    
    \STATE \textbf{Initialise parallel path index}: $n=1$.
    
    \STATE \textbf{Set the first active column (as one of the top $P$ candidates obtained previously)}: 
    $\mathbf{a}_{\hat{m_1}} = \mathbf{a}_{j_n}$.

    \STATE \textbf{Run MLMP algorithm and find the remaining $K-1$ active columns}:
    Denote the detected codeword as $\hat{\mathbf{s}}_n$. 

    \STATE \textbf{Update:}  $n=n+1$; if $n \leq P$ go back to step 5.


    \STATE \textbf{Select best among the $P$ codewords based on ML metric from }(\ref{ML_det}): $\hat{\mathbf{s}} = \argmax_{1 \leq n \leq P} p(\mathbf{y}_1,\cdots,\mathbf{y}_D | \hat{\mathbf{s}}_n)$.

\end{algorithmic}
\end{algorithm}

\subsection{Performance Guarantees of MLMP Algorithm}
In this section, we describe some important sparse recovery bounds for the MLMP algorithm for non-coherent decoding of SPARCs in SIMO flat-fading channels.

\begin{lemma}
For the non-coherent SIMO flat fading channel, the MLMP coincides with the maximum likelihood detector for SPARCs when sparsity $K=1$.
\end{lemma}

\begin{proof}
When a SPARC codeword is formed with $K=1$, the codeword is equal to one of the columns of the dictionary matrix $\mbf{A}$. Hence, the maximum likelihood detector for SPARCs with one section reduces to, 
\begin{align}
    \hat{\mathbf{s}} &= \argmax_{\mbf{s}} \sum_{i=1}^{D} |\langle \mathbf{y}_i,\mathbf{s} \rangle|^2, \\
    &= \argmax_{\mbf{a}_m \in \mathcal{A}} \sum_{i=1}^{D} |\langle \mathbf{y}_i,\mathbf{a}_m \rangle|^2,
\end{align}
since, $\beta_{\mbf{s}}$ and $\gamma_{\mbf{s}}$ are equal for all $\mbf{a}_m \in \mathcal{A}$. This coincides with the first iteration of the MLMP algorithm, and since $K=1$, we run the MLMP algorithm for only 1 iteration. 
\end{proof}

We now consider the sparse recovery guarantee of the MLMP algorithm in the absence of additive complex white gaussian noise (CWGN). 
\begin{theorem} \label{Theorem1}
    For SPARC codes defined on a dictionary matrix $\mbf{A}$ with mutual coherence $\mu$, the MLMP decoder with fixed $\beta=1$ and $\gamma=0$ for all iterations, recovers the support perfectly in the absence of CWGN if, 
    \begin{equation} \label{K bound}
        K \leq \frac{1+\mu}{2\mu}.
    \end{equation}
\end{theorem}
\begin{proof}
The proof of Theorem \ref{Theorem1} is given in Appendix \ref{proof_thm1}. \end{proof}
We observe that this bound coincides with that of orthogonal matching pursuit (OMP) for $K$-sparse signals \cite{OMP_bound}. MUB matrices of size $N \times N^2$ have mutual coherence $1/\sqrt{N}$, which, according to the bound, means that MLMP works for $K < 1 + \frac{\sqrt{N}}{2}$. The code rate of SPARC approaches zero as $N \rightarrow \infty$. Hence, only SPARCs with small code rates can be supported with greedy decoders at large block lengths.

\subsection{Modified Block-OMP (MBOMP)} 
\label{mbomp_section}
In \cite{ji2019pilot}, SPARC in unknown SIMO fading channels was studied with a block-OMP (BOMP) \cite{eldar2010block} based decoder. The BOMP algorithm is based on the successive cancellation method, in which it estimates unknown fading channels based on the detected columns and then \emph{removes} their contribution from the observation to compute the residual. For the flat-fading model, the modified-BOMP (MBOMP) algorithm is obtained by improving the channel estimation step as given below. The channel estimate $\hat{h}_{i,k}$ and the residual $\mathbf{r}_i^{(k)}$ based on the columns detected so far $\{\mathbf{a}_{\hat{m}_1},\cdots,\mathbf{a}_{\hat{m}_{k}}\}$ are computed as,
\begin{align} 
\hat{h}_{i,k} = \frac{\langle \mathbf{y}_{i}, \sum_{\ell=1}^{k} \mathbf{a}_{\hat{m}_\ell}\rangle }{\| \sum_{\ell=1}^{k} \mathbf{a}_{\hat{m}_\ell}\|^2}, \hspace{1mm} \mathbf{r}_i^{(k+1)} = \mathbf{y}_{i} - \hat{h}_{i,k} \sum_{\ell=1}^{k} \mathbf{a}_{\hat{m}_\ell}. \nonumber
\end{align}
The next active column is selected as $$\mathbf{a}_{\hat{m}_{k+1}} = \argmax_{\mathbf{a}_{m} \in \mathcal{A} \setminus \{\mathcal{\hat{A}}_{1} \cup \cdots \cup \mathcal{\hat{A}}_{{k}} \}} \sum_{i=1}^D |\langle \mathbf{r}_i^{(k+1)},\mathbf{a}_m\rangle|^2.$$ 
This OMP-based algorithm is susceptible to channel estimation errors, especially at low signal-to-noise ratio levels. On the other hand, MLMP in ~(\ref{MLMP_kth_iter}), which is based on a non-coherent ML metric and successive combining method, does not involve any channel estimation.



\section{Approximate message assing for simo sparcs (SAMP)} \label{AMP section}

SPARCs have gained significant interest as a promising alternative for future wireless technologies, due to the fact that they are capacity-achieving. This was proven in \cite{2017_Crush_unmod_sparc} for SPARCs communicating over the AWGN channel, using the approximate message passing (AMP) decoder\cite{2009_Donoho_CS_AMP,2013_Donoho_BlockAMP,2021_Schniter_LAMP}. Specifically, the authors proved that at asymptotic block lengths, SPARCs can achieve arbitrarily low section error rates for code rates $R<C$. Similarly, the authors in \cite{2017_barbier_AMP_SPARCs} have also proved the capacity-achieving nature of SPARCs in AWGN channel using spatially coupled design matrices and rigorously analysed them using statistical-physics-based methods like the replica method. However, the empirical performance of SPARCs at short block lengths is far from capacity. In this regard, we introduce SPARC-AMP (SAMP) for non-coherent decoding of short block length SPARCs over i.i.d. flat-fading SIMO channels, and compare it against greedy methods like MLMP and MBOMP. We use the MMSE estimator for a SPARCs through the SIMO channel as the denoising function in SAMP.

The SIMO flat fading channel model in (\ref{MSPARC_MMV}) can also be represented in matrix form as, 
\begin{equation} 
    \mathbf{Y = Axh}^T+\mbf{V = AG + V},
    \label{MMV system model}
\end{equation} 
where the columns of $\mbf{Y} \in \mathbb{C}^{N\times D}$ are the observations at the multiple antennas of the receiver, and $\mbf{G}\in \mathbb{C}^{L\times D}$ is the true row-sparse matrix. The AMP decoder iteratively generates estimates of $\mathbf{G}$ denoted by $\{\hat{\mathbf{G}}^t\}_{t > 0}$, from $\mbf{Y}$. By setting the initial $\hat{\mathbf{G}}^0$ to an all-zero matrix, the SAMP recursion goes as follows:
\begin{align} 
    \mathbf{Z}^t &= \mathbf{Y - A}\hat{\mathbf{G}}^t + \mbf{O}^{t-1}, \label{AMP1} \\
    \mathbf{B}^t &= \hat{\mathbf{G}}^{t} + \mathbf{A^*Z}^t, \label{AMP2}\\
    \hat{\mathbf{G}}^{t+1} &= \eta_t(\mbf{B}^t ;\tau_t^2). \label{AMP3}
\end{align}
where quantities with negative indices are set to zero. The functions $\{\eta_t(\cdot;\tau_t^2)\}_{t\geq 0}$ act as denoising functions. The last term of (\ref{AMP1}), i.e., $\mbf{O}^{t-1}$, is known as the \emph{Onsager correction term} which is calculated as,
\begin{equation}
    \mbf{O}^{t-1} = \frac{L}{N}\mathbf{Z}^{t-1} \langle\!\langle \eta'_{t-1}(\mathbf{A^*Z}^{t-1} + \hat{\mbf{G}}^{t-1};\tau_{t-1}^2)\rangle\!\rangle,\label{onsager_term}
\end{equation}
where $\eta'_t(\cdot;\tau_t^2)$ denotes the derivative of the denoising function with respect to its first argument, and $\langle\!\langle \cdot \rangle\!\rangle$ denotes the average of the input argument.
Due to the presence of the onsager term, the input to the denoising function is assumed to be distributed as follows:
\begin{equation}
   \mbf{B}^t = \hat{\mbf{G}}^{t} + \mathbf{A^*Z}^t \approx  \mathbf{G} + \tau_t \mathbf{W} ,
\end{equation}
with $\mathbf{G}$ being the true row-sparse matrix, and each element of $\mbf{W}$ distributed as $\mathcal{CN}(0,1)$. This implies that the input to the denoiser is distributed as the true signal corrupted by a complex gaussian noise with variance $\tau_t^2$. These $\{\tau_t^2\}_{t\geq0}$ are known as \emph{state evolution} (SE) parameters. This distribution is crucial for the design of the MMSE estimator and AMP convergence, and is possible due to the Onsager correction term in (\ref{AMP1}).

\subsubsection{MMSE Denoiser}
We use an MMSE estimator as the denoising function in (\ref{AMP3}) and it produces a row-wise output from its input $\mbf{B}^t$ as follows,
\begin{equation}
    \hat{\mbf{G}}_{k,:}^{t+1} = \mathbb{E}[ \mbf{G}_{k,:} | \mathbf{B}_{\sec(k),:}^t = \mbf{G}_{\sec(k),:} + \tau_t\mbf{W}_{\sec(k),:} ], \forall k \in [L].
\end{equation}
This makes $\hat{\mbf{G}}^{t+1}$ the optimal Bayes estimate from $\mbf{B}^t$. For $k \in [L]$, we have 
\begin{align}
    \hat{\mbf{G}}_{k,:}^{t+1} &=  \mathbb{E}[\mbf{G}_{k,:} | \mbf{B}_{\sec(k),:}^t ], \nonumber \\ 
    &= \sum_{\ell \in \sec(k)} \mathbb{E}[\mbf{G}_{k,:}|\mbf{B}_{\sec(k),:}^t,e_\ell]p(e_\ell | \mbf{B}_{\sec(k),:}^t), \\
    &= \mathbb{E}[\mbf{G}_{k,:}|\mbf{B}_{\sec(k),:}^t,e_k]p(e_k | \mbf{B}_{\sec(k),:}^t) \label{mmse_krow},
\end{align}
where $e_k$ is the event that $\mbf{G}_{k,:}$ is the only nonzero row in $\mbf{G}_{\sec(k),:}$. Given $e_k$, the entries of $\mbf{B}_{k,:}^t$ are distributed as, 
\begin{equation}
    \mbf{B}_{k,:}^t = \mbf{h}^T + \tau_t\mbf{W}_{k,:}. \label{B_e_k}
\end{equation}
Given that $\mbf{h} \sim \mathcal{CN}(0,\sigma_h^2 \mathbb{I}_d)$, the expectation in (\ref{mmse_krow}) can be computed as 


\begin{equation}
    \mathbb{E}[\mbf{G}_{k,:}|\mbf{B}^t_{\sec(k),:},e_k] = \frac{\sigma_h^2}{\sigma_h^2 + \tau_t^2}\mbf{B}^t_{k,:}. \label{exp_value}
\end{equation}
Using the Bayes' rule, we can compute \begin{equation}
    p(e_k|\mbf{B}_{\sec(k),:}^t)  = \frac{p(\mbf{B}_{\sec(k),:}^t|e_k) p(e_k)}{\sum_{\ell \in \sec(k)} p(\mbf{B}_{\sec(k),:}^t|e_\ell)p(e_\ell)},\label{bayes e|b} 
\end{equation}
but $p(e_k)=p(e_\ell), \forall \ell \in \sec(k)$, as the location of the only non-zero entry in each section is equiprobable.

Given $e_k$ and for $\ell \in \sec(k)$ from (\ref{B_e_k}) we have
\begin{equation} 
p (\mbf{B}_{\ell,:}^t|e_k) \sim
\begin{cases}
\mathcal{CN}(0,(\sigma_h^2 + \tau_t^2)\mathbb{I}_D), & \text{for} \quad  \ell=k, \\
\mathcal{CN}(0,\tau_t^2\mathbb{I}_D), & \text{else} \quad \ell \neq k,
\end{cases} \label{dist s|e}
\end{equation} 
and, 
\begin{equation}
    p(\mbf{B}_{\sec(k),:}^t|e_k) = \prod_{\ell\in \sec(k)} p(\mbf{B}_{\ell,:}^t|e_k). \label{B_prod}
\end{equation}
Finally, since $\mbf{G}$ and $\mbf{W}$ are independent and all entries of $\mbf{h}^T$ are i.i.d, from (\ref{dist s|e}) and (\ref{B_prod}) we have,
\begin{align}
    p(\mbf{B}_{\sec(k),:}^t|e_k) & \propto \exp{\left (\frac{-\|\mbf{B}_{k,:}^t\|^2}{\tausig} + \sum_{\ell \in \sec(k)\setminus k} \frac{-\|\mbf{B}_{\ell,:}^t\|^2}{\tau_t^2} \right )}, \nonumber\\
    &\propto \exp{ \left (\frac{\sigma_h^2 \|\mbf{B}_{k,:}^t\|^2}{\tau_t^2(\tausig)} - \sum_{\ell \in \sec(k)} \frac{\|\mbf{B}_{\ell,:}^t\|^2}{\tau_t^2} \right )}. \label{joint_pdf}
\end{align}

Plugging the density functions (\ref{joint_pdf}) along with (\ref{bayes e|b}) and (\ref{exp_value}) in (\ref{mmse_krow}), we get the MMSE estimator for (\ref{AMP3}) as,
\begin{equation}
\eta_t(\mbf{B}_{k,:}^t;\tau_t^2) = \frac{\exp{\left ( \frac{\sigma_h^2\|\mbf{B}_{k,:}^t\|^2}{ \tau_t^2 (\tausig)}  \right )}}{\sum_{\ell\in \sec(k)} \exp{\left ( \frac{ \sigma_h^2\|\mbf{B}_{\ell,:}^t\|^2}{\tau_t^2 (\sigma_h^2 + \tau_t^2)}  \right )}} \frac{\sigma_h^2}{\tausig }  \mbf{B}_{k,:}^t. \label{MMSE estimator}
\end{equation}
The MMSE denoiser acts row-wise on $\mbf{B}^t$. Since the denoiser takes a row vector of $\mbf{B}^{t}$ as input and produces a row vector output, the derivative $\eta_{t}'(\cdot;\tau_t^2)$ with respect to its input row vector in (\ref{MMSE estimator}) is a Jacobian matrix. Specifically $(i,j)$\textsuperscript{th} entry of the Jacobian matrix corresponding to the derivative $\eta_t$ with the row $\mbf{B}_k^t$ is,
\begin{equation}
    \frac{\partial \eta(\mbf{B}_{k,i}^{t})}{\partial \mbf{B}_{k,j}^{t}} =  \eta (\mbf{B}_{k,i}^{t}) \left [\frac{\delta_{ij}}{\mbf{B}_{k,i}^{t}} + \frac{\bar{\mbf{B}}_{k,j}^{t}\sigma_{h}^{2}}{\tau_{t}^{2} (\tausig)} + \frac{\eta (\mbf{B}_{k,j}^t)\bar{\mbf{B}}_{k,j}^t}{\tau_t^2 \mbf{B}_{k,j}^t}  \right],
\end{equation}
where, $\mbf{B}_{k,j}^t$ corresponds to the $j$\textsuperscript{th} entry of the $k$\textsuperscript{th} row of $\mbf{B}^t$, $\bar{\mbf{B}}_{k,j}$ is the complex conjugate of $\mbf{B}_{k,j}$ and $\delta_{ij}=1$ only if $i=j$. This gives a square matrix for every row input, and the Onsager term in (\ref{onsager_term}) is calculated as the average of all Jacobian matrices corresponding to all rows of $\mbf{B}^t$, giving a $D\times D$ matrix.

\subsubsection{State Evolution Parameters}
The constants $\{\tau_t^2\}_{t\geq0}$ in (\ref{AMP3}) predict the mean squared error (MSE) between the true $\mathbf{G}$ and the estimated signal $\hat{\mbf{G}}^t$ at each iteration $t$ and are known as state evolution (SE) parameters. The SE recursion goes as follows:
\begin{align}
    \tau_{t+1}^2 &= \sigma_v^2 + \frac{1}{ND} \mathbb{E} [\|\hat{\mbf{G}}^{t+1} - \mathbf{G} \|^2], \label{SE1}  \\
    &= \sigma_v^2 + \frac{1}{ND} \mathbb{E} [\|\eta_t(\mathbf{G} + \tau_t\mathbf{W}) - \mathbf{G} \|^2]. \label{SE2}
\end{align}
with initial $\tau_0^2 = \sigma_v^2 + \frac{1}{ND} \mathbb{E}[\| \mathbf{G}\|^2]$. If AMP is run for $T$ iterations, the values of $\tau_t^2,\forall t \in [T]$ are pre-computed from (\ref{SE2}) using Monte Carlo simulations. As we see, realisations of true signal are a part of this computation, which makes it clear that $\tau_t^2$ computation is not part of the actual decoding; rather, the pre-computed values for a given $\sigma_v^2$ are used at the time of actual decoding. We call these values the \emph{offline} state evolution parameters. 

\begin{figure}
    \centering
    \includegraphics[width=\columnwidth]{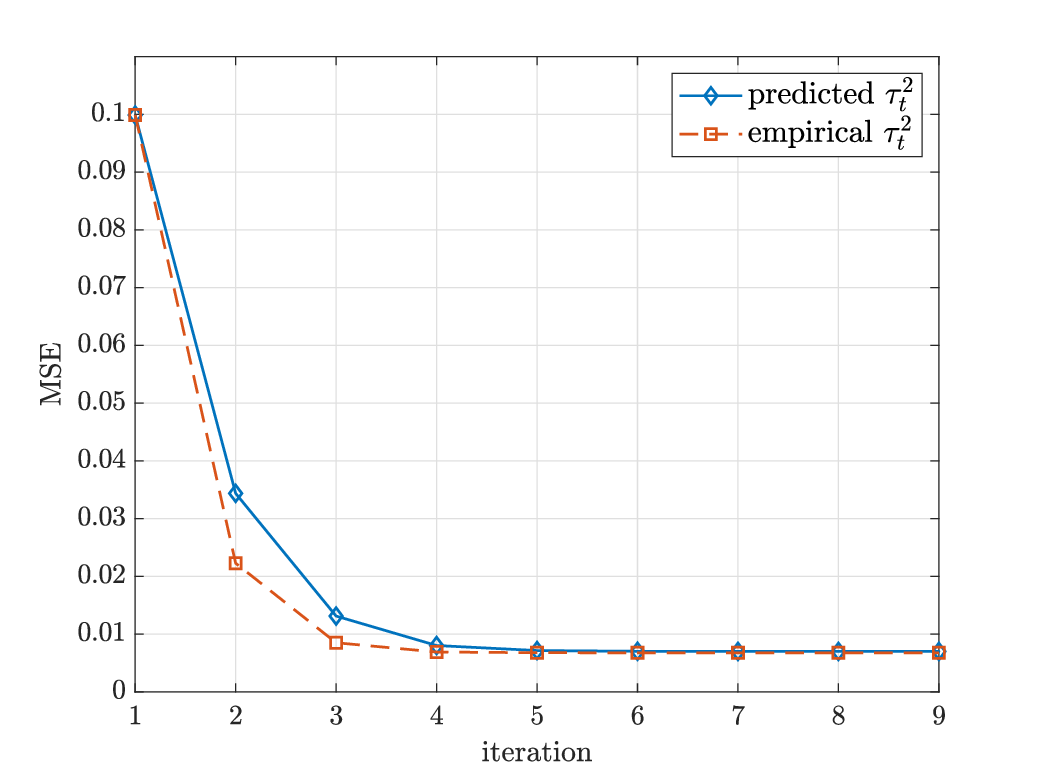}
    \caption{Comparison of state evolution parameters and empirical MSE for (256,124) SPARCs at Eb/N0 = 18 dB.}
    \label{fig:tau_emp_pred}
\end{figure}

In (\ref{SE1}), the average MSE between the actual estimate $\mbf{\hat{G}}$ and the true $\mbf{G}$ is computed across many decoding realisations to give us the empirical SE. Figure \ref{fig:tau_emp_pred} shows the comparison of the predicted SE (\ref{SE2}) parameters and the empirical SE (\ref{SE1}) with iteration number $t$. The SE parameters closely track the empirical MSE and converge to a fixed point. The number of iterations $T$, after which we terminate the SAMP algorithm, can be equal to the number of iterations for the offline predicted SE parameters to converge to the fixed point.

As an alternative approach, we can compute the SE parameters $\tau_t^2$ while decoding the message. This is called the \emph{online} state evolution \cite{2017_greig_Techniques}. At each step $t$, an online estimate of $\tau_t^2$ is calculated as,
\begin{equation}
    \Tilde{\tau}_t^2 = \frac{\|\mbf{Z}^t\|^2}{ND}.
\end{equation}
In \cite{2017_Crush_unmod_sparc,2017_Crush_err_exponent} the authors proved that with large $N$, $\Tilde{\tau}_t^2$ is close to $\tau_t^2$ with high probability. Similar online estimates for AMP have been used in \cite{2021_Ramji_mod_sparc,2017_Crush_err_exponent}. This provides an alternative and computationally simpler way to estimate the state evolution parameters. With online SE computation, we can run SAMP until the SE evolution parameter stops changing with iterations. Later, we will show through simulations that SAMP with online state evolution performs very similarly to SAMP with offline SE computation.  

\section{Existing Error Control Codes} \label{Polar codes section}

\subsection{Polar Codes with Pilot-Aided Transmission}
Pilot-aided transmission (PAT) enables the use of conventional error-correcting codes by transmitting known pilot symbols and facilitating channel estimation. The receiver can decode the data based on the estimated channel \cite{apk1,apk2}. In our flat fading model, a single pilot symbol with energy $E_p = \alpha E$ is sufficient to estimate the unknown fading channel, and the remaining energy $E_d=(1-\alpha)E$ is allocated for the data codeword. The value of $\alpha$ can be chosen to maximise the effective SNR, as described below. The observations at the $i^{th}$ receive antenna corresponding to the pilot and the data are given respectively as $ y_{p,i} = h_i\sqrt{E_p} + v_{p,i}$ and $\mathbf{y}_{d,i} = h_i\sqrt{\frac{E_d}{N}}\mathbf{s}_d + \mathbf{v}_{d,i}$.
The MMSE estimate of the channel is $\hat{h}_{i} = \frac{\sqrt{E_p}\sigma_h^2}{E_p \sigma_h^2 + \sigma_v^2} y_{p,i}$, and the estimation error is $\Tilde{h}_i = h_i - \hat{h}_i$. Now, the data observation can be written as,
\begin{equation} \label{y_d after estimation}
    \mathbf{y}_{d,i} = \hat{h}_i\sqrt{\frac{E_d}{N}}\mathbf{s}_d + \Tilde{h}_i\sqrt{\frac{E_d}{N}}\mathbf{s}_d + \mathbf{v}_{d,i}.
\end{equation} 
The last two terms in (\ref{y_d after estimation}) act as noise, which includes the CWGN noise and the channel estimation error. We choose the value of $\alpha$, in order to maximise the effective SINR as $ \displaystyle \alpha_{opt} = \argmax_{0<\alpha<1} \frac{\frac{E_d}{N}\hat{\sigma}_h^2}{\frac{E_d}{N}\Tilde{\sigma}_h^2 + \sigma_v^2}$,
where $\Tilde{\sigma}_h^2 = \frac{\sigma_h^2\sigma_v^2}{E_p\sigma_h^2 + \sigma_v^2}$ and $\hat{\sigma}_h^2 = \sigma_h^2 - \Tilde{\sigma}_h^2$. 
Once we have the estimate of the unknown fading channel, we perform \emph{maximal ratio combining (MRC)} of the observations ${\mathbf{z}}_d = \sum_{i=1}^{D} \hat{h}_i^{*}\mathbf{y}_{d,i}$ and decode the codeword $\mathbf{s}_d$.

\renewcommand{\tabcolsep}{1pt}
\begin{table*}[t]
\caption{Computational complexity comparison} 
\resizebox{\textwidth}{!}{%
\begin{tabular}{|p{0.4in}|p{2.5in}|p{2.5in}|p{2in}|} 
\hline
 \scriptsize
\textbf{Decoder} &  \scriptsize \textbf{Additions}                 &  \scriptsize \textbf{Multiplications}      \\ \hline
 \scriptsize MBOMP  &  \scriptsize $LD(K+1)(N+1)/2$ $+7NK-2K$  &   \scriptsize $LD(K+1)+7NK$  \\ \hline
 \scriptsize MLMP &  \scriptsize $L(K+1)(3N + n+3D)/2 + DLN + K+1$ &   \scriptsize $L(K+1)(N+n+D+3) + K+2$    \\ \hline
 \scriptsize PAT-Polar &  \scriptsize $2N(4n+18) + J(2N \log_2(2N)(2n+6)+2N(2n+4) - J) + (2N+1)(2N-N_b+1)$ &  \scriptsize
  $2N(4n+18) + J(2N\log_2({2N})(3(n-1)+2) + 2N(2+4(n-1))) + (2N+1)(2N-N_b+1)$  \\ \hline
  

 \scriptsize SAMP &  \scriptsize $T(LND+2ND+4L+KM-K) + (T-1)(5LD^2-D^2-2LD+LND+N)$  &  \scriptsize $T(3LD+11L)+(T-1)(8LD^2+4LD+1)$ \\ \hline

\end{tabular}%
}

\label{comp_complexity}
\end{table*}

\subsection{Noncoherent Polar Codes}
The authors in \cite{yuan2021polar} proposed a polar-coded non-coherent (NC-Polar) communication without transmitting pilots.  First, the magnitude of the unknown fading channel is estimated based on the norm of the received signal. The unknown phase is then found using a grid search using SCL polar decoding for each hypothesised value of $\theta$. However, this method faces a phase ambiguity between $\theta$ and $\theta+\pi$, which is resolved by using a CRC check. To achieve maximal ratio combining, all possible combinations of $h_i$ and $-h_i$ must be considered, leading to a significant increase in computational complexity with the number of antennas.

\subsection{Complexity Comparison}
The exact number of additions and multiplications required for each technique is given in Table \ref{comp_complexity}. For exponentiation and logarithms, we use an order-$5$ polynomial approximation. For the AMP decoder, we use $T=10$ iterations. We use an SCL decoder with list size $J=16$ for polar codes. For the non-coherent polar decoding, due to the grid-search-based phase estimation, the complexity linearly increases with the number of grid points. For the code rate $(2N, N_b)=(128,40)$, with $D=4$ antennas, we numerically calculate the total number of flops $F$ for each technique and normalise them by the number of flops of MLMP, resulting in the relative scale $F^{MBOMP}: F^{MLMP}: F^{PAT-Polar}: F^{NC-Polar}: F^{SAMP}\approx 0.5:1:0.1:8:4.8$.

\subsection{Coherent Sphere Packing Bound} \label{SPB section}
In \cite{shannon1959probability}, Shannon proved that the probability of error for codes transmitted over the AWGN channel is lower bounded by the sphere packing bound, which uses the cdf of \emph{noncentral t-distribution}. Now, for coherent communication over a SIMO channel (\ref{MSPARC_MMV}) with $D$ receive antennas, we combine the received codewords using maximal ratio combining (MRC). 
\begin{align}
    \mbf{y} &= \frac{1}{\|\mbf{h}\|}\sum_{i=1}^{D}h_i^*\mbf{y}_i, \nonumber \\
    \mbf{y} &= \|\mbf{h}\|\mbf{s} + \mbf{v}.
\end{align}
Now, this is similar to the AWGN channel with power $\alpha P$, where $\alpha = \sum_{i=1}^Dh_i^2$, and $P$ is the average power constraint of codewords $\mbf{s}$. Since each $h_i \sim \mathcal{CN}(0,1)$, we have $\alpha \sim Gamma(D,1)$. To find the coherent sphere packing bound, we average the probability of error with respect to the distribution of $\alpha$. Hence, 
\begin{equation} \label{cohernet_spb}
    P_e = \int_0^\infty \text{nctcdf}(\sqrt{2N-1}\cot\theta,\sqrt{2N\alpha P},2N-1)\text{f}(\alpha)d\alpha .
\end{equation}
where nctcdf denotes the CDF of the noncentral t-distribution and $\text{f}(\alpha)$ denotes the gamma distribution of $\alpha$. $\theta$ refers to the half-angle of the cone corresponding to a codeword, formed by equally partitioning the surface of a unit ball in $R^{2N}$ by all possible codewords.

\section{Simulation results} \label{Results section}
We study the BLER performance of MLMP and SAMP and compare it with other existing sparse recovery algorithms and error control coding techniques. To generate SPARC codewords, we employ complex MUB matrices \cite{woottersMUB} with $N=64$ and $128$ rows denoted by MUB64 and MUB128, respectively, with unit norm columns. These matrices have $L=N^2$ columns, which are then partitioned using a partitioning algorithm, to obtain a certain code rate $R$. SPARC codewords are generated using a sparse sum of $K$ columns of these matrices, generating complex codewords of size $N$, equivalent to $ 2N$ real channel uses. The expected codeword energy is $E_s = K$, with $E_b = E_s/N_b$ and $N_b$ denoting the number of information bits transmitted. Our simulations plot BLER versus $E_b/N_0$, where $N_0 = \sigma_v^2$ for complex codewords. The $(2N, N_b)$ SPARC scheme transmits $N_b$ bits over $2N$ real dimensions, with a rate of $N_b/2N$ bits per real dimension (bpcu). With $D$ receive antennas, the channel variance is normalised to $\sigma_h^2 = 1/D$, ensuring unity array gain, so performance improvements are solely due to diversity gain.

\subsection{Comparison with Greedy Algorithms}

\subsubsection{Comparison with Block-OMP}
In Figure \ref{fig:mlmp_vs_bomp}, we compare the performance of MLMP with the modified block-OMP (MBOMP) algorithm, introduced in Section \ref{mbomp_section} as a modification to the existing block-OMP (BOMP) algorithm \cite{2023_kumar_PBOMP,eldar2010block}. Since the system is flat-fading in all the $D$ channels, the nonzero rows of the row-sparse matrix $\mbf{G}$ in (\ref{MMV system model}) are equal to these fading coefficients. MBOMP takes advantage of this fact to calculate only one coefficient for each sparse column of $\mbf{G}$, utilising columns from all sections, thereby giving a better estimate. MBOMP gives about $0.5$ dB improvement in BLER over traditional BOMP.

\begin{figure}
    \centering
    \includegraphics[width=\columnwidth]{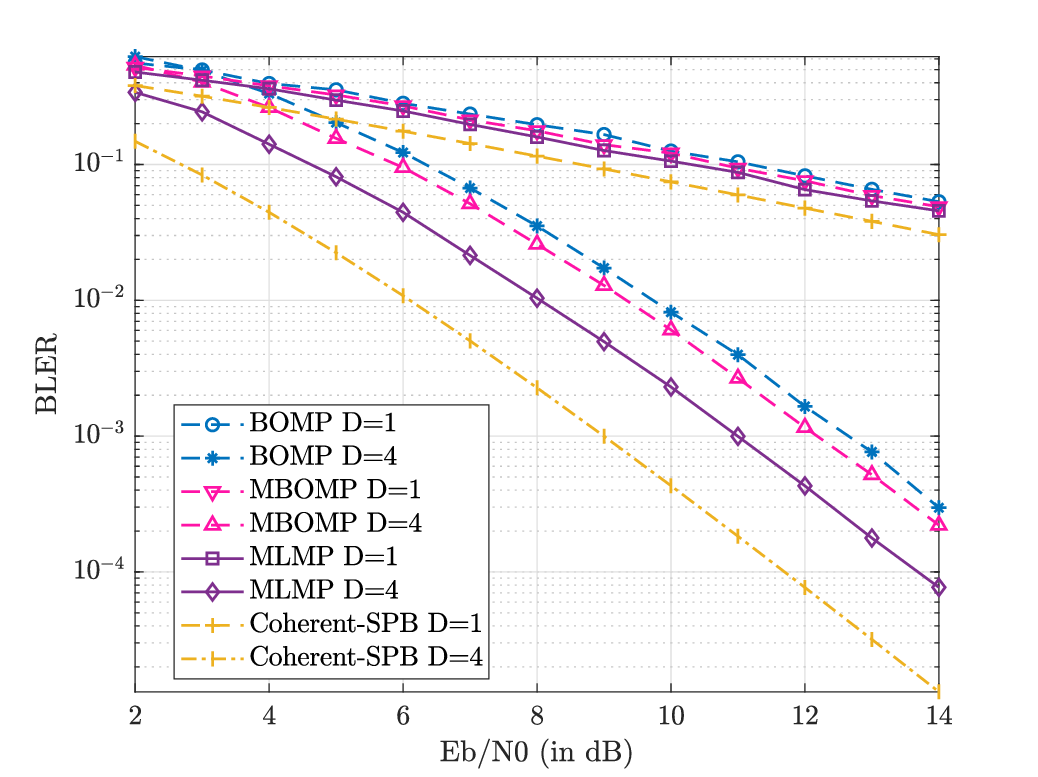}
    \caption{Comparison of MLMP with existing sparse recovery algorithms}
    \label{fig:mlmp_vs_bomp}
\end{figure}

Unlike in OMP-based methods, there is no additional step of nonzero coefficient estimation in MLMP; instead, it computes the ML metric with one column at a time and considers that column, which, when included, gives the highest metric. In addition, when OMP projects the residuals onto the orthogonal complement of the detected columns, the signal component of the undetected columns is reduced. In contrast, MLMP is based on a successive combining method where the previously detected columns are still part of the metric. For these reasons, we observe a significant improvement of up to 1.2 dB with MLMP over MBOMP for (128,40) SPARC and 4 receive antennas. We note that all three algorithms shown in the plot use 8-parallel paths. In the exact figure, we also compare the performance of these sparse recovery algorithms with the coherent sphere packing bound (\ref{cohernet_spb}), which presents a lower bound on the BLER for a particular code rate when the codeword is passed through a SIMO channel with perfect CSI.

\subsubsection{Effect of Parallelisation}
Parallelisation can enhance the performance of the MLMP decoder, but the effect differs with the code rates, that is, the sparsity $K$. As shown in (\ref{K bound}), MLMP can guarantee perfect recovery up to a certain $K$ in the noiseless environment, depending on the mutual coherence $\mu(\mbf{A})$ of the dictionary matrix. Although perfect recovery is not feasible in a noisy environment, we can get good BLER performance with a high $K$ using parallel paths. If $K$ is small enough, then MLMP with a single path performs well, and introducing parallel paths does not significantly improve the performance over $1-$MLMP. However, as $K$ increases to support higher code rates, the first iteration of MLMP is prone to more errors, which can be mitigated by using parallel paths. This phenomenon is shown in Figure \ref{fig:PMLMP_K478}, where we compare the performance of $1-$MLMP vs $8-$MLMP at different code rates or sparsity levels for the same $N$. We observe that the improvement from using parallel paths is more significant for $(128,72)$ obtained from $K=8$ sections, than for $(128,40)$ obtained from $K=4$ sections of MUB64. Another interesting observation is that MLMP with 16 paths only has a marginal improvement over 8-path decoders, showing diminishing returns in BLER as the number of parallel paths increases. This phenomenon is observed at all code rates. 

\begin{figure}
    \centering
    \includegraphics[width=\columnwidth]{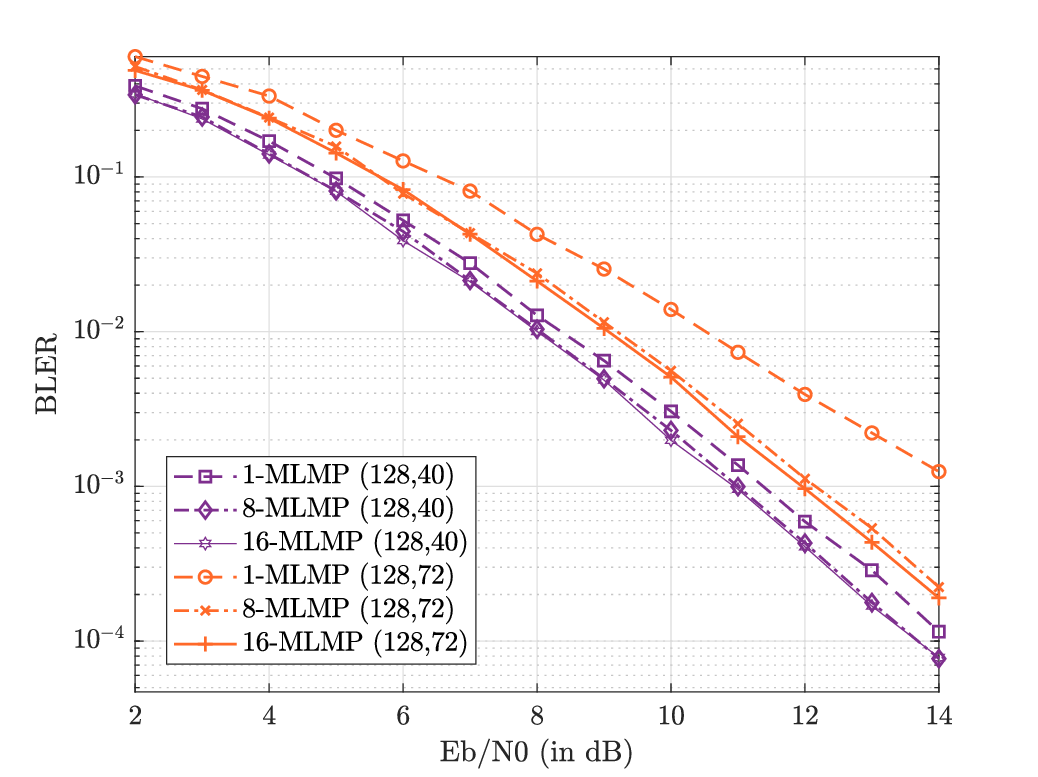}
    \caption{Parallel-MLMP at different code rates}
    \label{fig:PMLMP_K478}
\end{figure}

\subsection{Performance of SAMP}
\subsubsection{SAMP vs MLMP}
In Figure \ref{fig:SAMP_coderate}, we compare the performance of SAMP with MLMP and MBOMP at different code rates. At low code rates, greedy methods outperform SAMP. However, the sparsity undersampling ratio threshold for SAMP is similar to that of convex optimisation methods \cite{2009_Donoho_CS_AMP}, which is higher than that of greedy methods. We see in the figure that as K increases, that is, the code rate increases, SAMP outperforms MBOMP. However, MLMP outperforms both SAMP and MBOMP in all these cases. Since parallelisation cannot be done with SAMP, we compare it with single-path counterparts of the greedy methods. 
\begin{figure}
    \centering
    \includegraphics[width=\columnwidth]{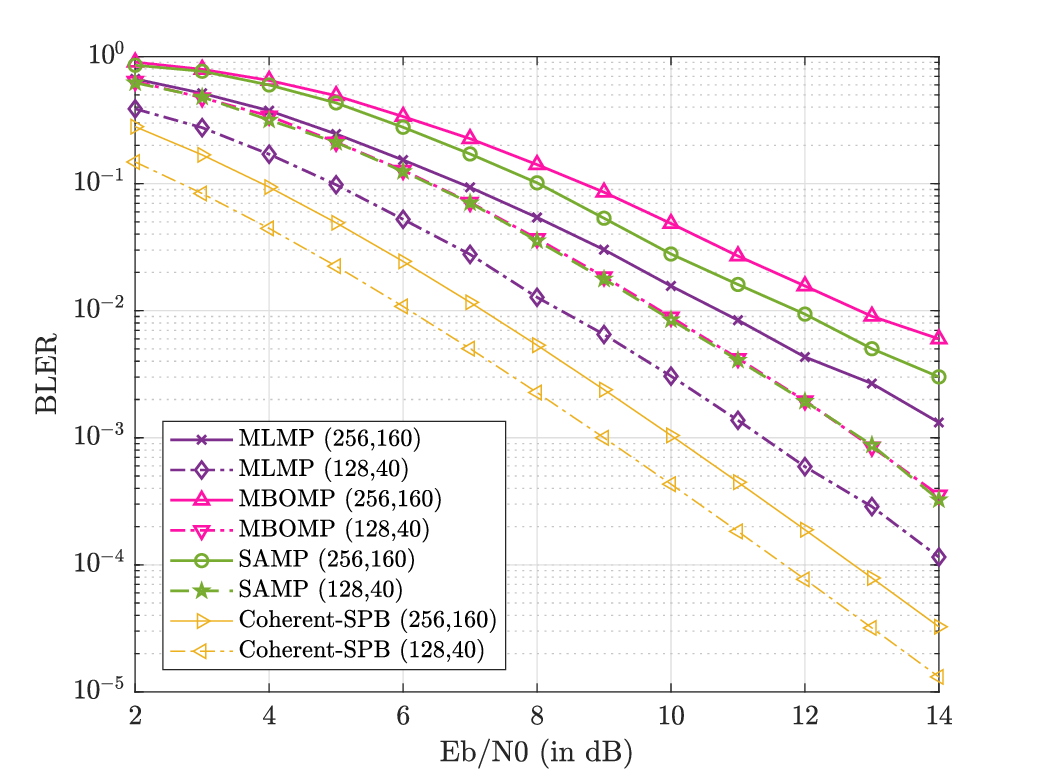}
    \caption{Performance of SAMP at increasing code rate}
    \label{fig:SAMP_coderate}
\end{figure}

\subsubsection{Performance of SAMP Online}
For SAMP, the state evolution parameters are first computed offline before actual decoding, and the number of iterations for AMP is fixed based on the number of iterations it takes for the state evolution to converge to a fixed point. Since offline computation of the state evolution parameters is computationally complex and has a large overhead, online computation during the decoding of the message is more suitable and practical. Figure \ref{fig:SAMP_online} shows the performance of SAMP online compared to SAMP offline. For the (256,124) SPARC code shown in the figure, SAMP takes six iterations for SE to converge to a fixed point at $14$ dB. Figure \ref{fig:SAMP_online} also shows the performance of SAMP when we let it run for 50 iterations, and shows no improvement over the original SAMP. This shows that AMP needs to run only until the SE converges. 

\begin{figure}
    \centering
    \includegraphics[width=\columnwidth]{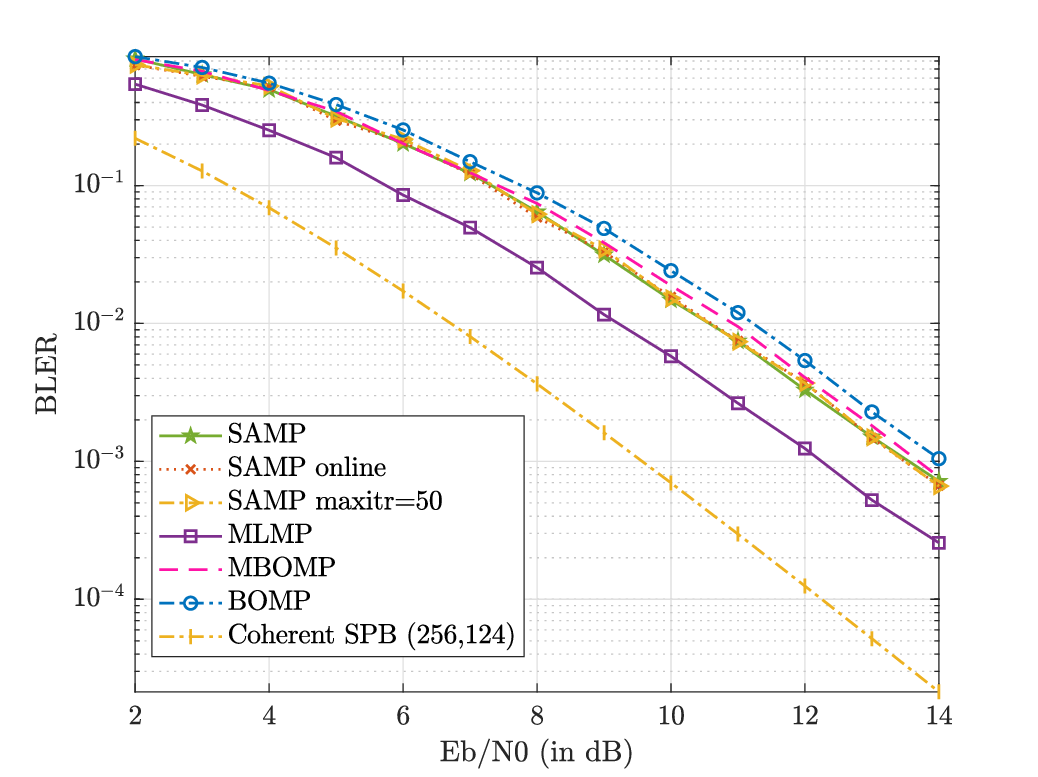}
    \caption{Performance of SAMP with online SE}
    \label{fig:SAMP_online}
\end{figure}


\subsubsection{Performance with different section sizes}
In SPARC, we can have a small number of sections with each section having a large number of columns, or we can have a large number of sections with each section having a small number of columns. With careful choice of parameters, we can have nearly the same coderate in both situations. In Figure \ref{fig:Gauss_perf}, we compare the performance of SAMP for approximately the same code rate, but obtained from different section sizes. Specifically, the code rate $(128,42)$ was obtained from very small section sizes of $64$ columns and $K=7$ sections, compared to $(128,40)$, which was obtained using large section sizes of $1024$ columns and $K=4$ sections. We observe that the performance of SAMP with higher $K$ is poor compared to that of lower $K$ with higher columns per section. We also observe the improvement we get by using MUB sequences for the dictionary matrix rather than a random Gaussian i.i.d. matrix distributed as $\mathcal{CN}(0,1/N)$.

\begin{figure}
    \centering
    \includegraphics[width=\columnwidth]{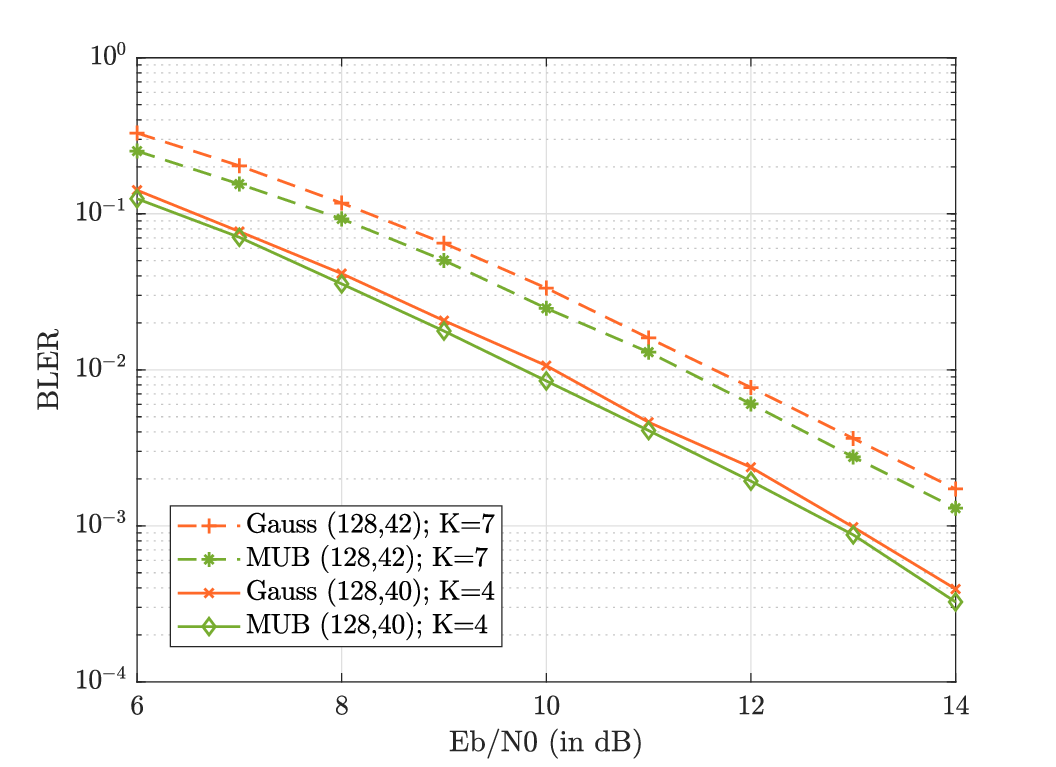}
    \caption{Performance of SAMP with different section sizes}
    \label{fig:Gauss_perf}
\end{figure}

\subsection{Comparison with Polar codes}
In Figure \ref{polar_comp}, we compare the performance of (130,40) polar code \cite{arikan2009polarcodes} used in the 5G NR standards along with pilot-aided transmission and list decoding with list size $J=8$ \cite{tal2015listdecoder}. We also show the performance of (128,40) polar codes for non-coherent detection \cite{yuan2021polar} using joint channel estimation and decoding, denoted by NC-Polar in the legend. With our proposed parallel MLMP decoder, we see that SPARC has significant gains over polar codes with the PAT strategy and non-coherent scheme \cite{yuan2021polar}.

\begin{figure}[t]
    \centering
    \includegraphics[width=\columnwidth]{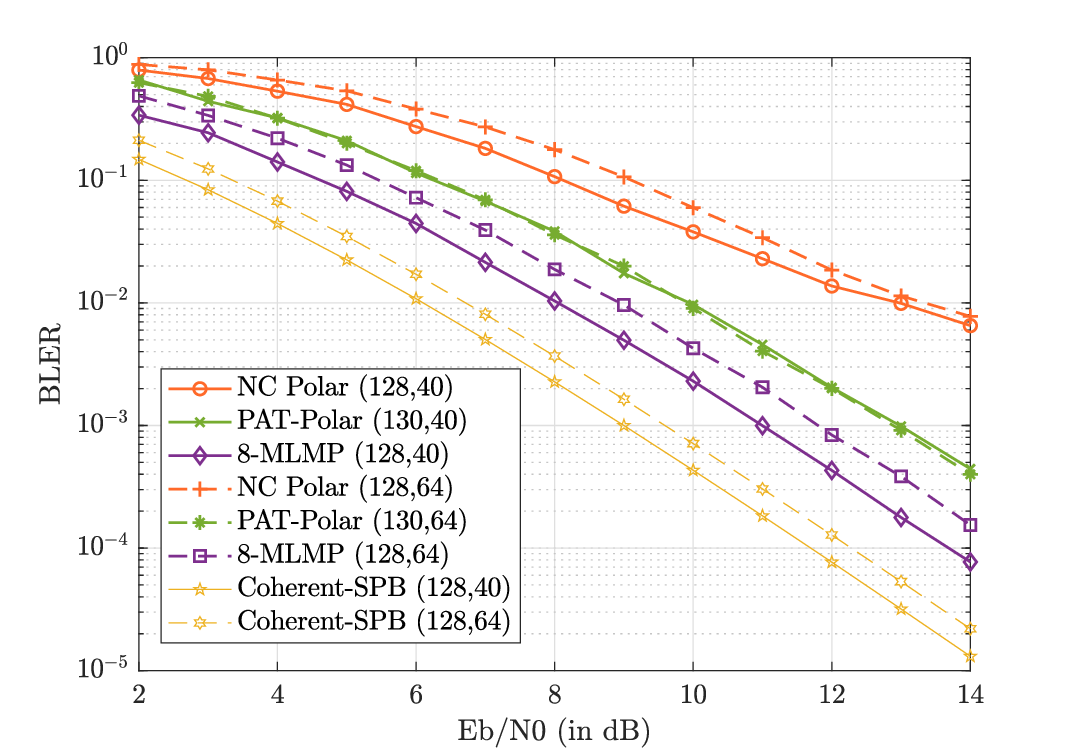}
    \caption{Comparison of MLMP with Polar codes}
    \label{polar_comp}
\end{figure}



\section{Conclusion} \label{conclusion}
In this work, we consider SPARCs with MUB dictionary matrix over an unknown flat-fading channel with multiple receive antennas. We developed a novel MLMP decoder, which was inspired by greedily reducing the search space of the non-coherent ML detector. We introduce the use of parallel paths for MLMP and existing greedy methods to improve the decoding performance. We derive the sparsity upper bound for MLMP for the noiseless system model and show that the MLMP recovery guarantees are the same as those of the state-of-the-art OMP algorithm. We also introduce the SPARC approximate message-passing (SAMP) algorithm for non-coherent detection and show that its performance can be accurately tracked using the state evolution formulation. In the short block length regime, MLMP significantly outperforms the SPARCs with Block-OMP-based and AMP-based decoders and PAT-aided polar codes and polar codes with non-coherent decoders. Although we focus on flat-fading SIMO channels in this paper, it is crucial to extend this work to MIMO frequency-selective channels in the future.

\appendices

\section{Proof of Theorem 1} \label{proof_thm1}

\begin{proof}
We can prove Theorem \ref{Theorem1} by showing that when the condition in (\ref{K bound}) is met, the ML metric corresponding to a chosen column will be higher than that of any other column that was not part of the codeword in each iteration. The noiseless observation can be represented as,
\begin{equation}
    \mbf{y}_i = h_i (\sum_{m \in \mathcal{S}}\mbf{a}_m), \forall i \in [D].
\end{equation}
Assume that $K$ satisfies the bound in (\ref{K bound}), and in the first $k$ iterations of MLMP, the $k$ columns of $\mbf{A}$ are correctly decoded. $\mathcal{S}'$ denotes the indices of correctly decoded columns in the first $k$ iterations, i.e., $\mathcal{S}' \subset \mathcal{S}, |\mathcal{S}'|=k$. Let us denote $\mbf{s}_k = \sum_{\ell \in \mathcal{S}'} \mbf{a}_{\ell}$, which represents the partially decoded codeword after $k$ iterations.

MLMP correctly decodes a column $\mbf{a}_{u}, u \in \mathcal{S}\setminus \mathcal{S}'$ if the minimum metric corresponding to any column $\mbf{a}_u, \forall u\in\mathcal{S}\setminus \mathcal{S}'$ is greater than the maximum value of a metric corresponding to a column $\mbf{a}_w,w\notin \mathcal{S}$.  
\begin{align}
    \sum_{i=1}^{D} |\langle h_i\mbf{s}, \mbf{s}_k + \mbf{a}_u\rangle|^2 &> \sum_{i=1}^{D} |\langle h_i\mbf{s}, \mbf{s}_k +  \mbf{a}_w\rangle|^2, \nonumber \\
     \left (\sum_{i=1}^{D} h_i \right) |\langle \mbf{s},\mbf{s}_k +  \mbf{a}_u\rangle|^2 &> \left (\sum_{i=1}^{D} h_i \right) |\langle \mbf{s}, \mbf{s}_k + \mbf{a}_w\rangle|^2, \nonumber \\
     |\langle \sum_{m \in \mathcal{S}}\mbf{a}_m, \mbf{s}_k + \mbf{a}_u\rangle| &>  |\langle  \sum_{m \in \mathcal{S}}\mbf{a}_m, \mbf{s}_k + \mbf{a}_w\rangle|. \label{mingmax}
\end{align} 
We have a common term on both sides given by $\langle \sum_{m\in \mathcal{S}}\mbf{a}_m,\mbf{s}_k \rangle$, whose real part, 
\begin{align}
    \alpha &= \Re\{\langle \sum_{m\in\mathcal{S}}\mbf{a}_m,\mbf{s}_k \rangle\}, \nonumber\\
    \mathllap{\Rightarrow\quad}  &\geq k - (Kk-k)\mu, \nonumber\\
    \mathllap{\Rightarrow\quad}  &\geq k (1- (K-1)\mu) > 0.
\end{align}
since, $K \leq \frac{1+\mu}{2\mu} < \frac{1+\mu}{\mu}$. 

From (\ref{mingmax}), considering that,
\begin{align}
    |\langle \sum_{m \in \mathcal{S}}\mbf{a}_m, \mbf{s}_k + \mbf{a}_u\rangle| &= |\underbrace{\langle \sum_{m \in \mathcal{S}}\mbf{a}_m,\mbf{s}_k\rangle}_{\alpha+j\zeta} + \langle \sum_{m \in \mathcal{S}}\mbf{a}_m,\mbf{a}_u \rangle|,  \nonumber \\
    \mathllap{\Rightarrow\quad} &\geq |\alpha + j\zeta + 1 + (K-1)\mu|, \nonumber \\
    \mathllap{\Rightarrow\quad} &\geq |(\alpha+1) + j\zeta| - (K-1)\mu. 
\end{align}
And,
\begin{align}
    |\langle \sum_{m \in \mathcal{S}}\mbf{a}_m, \mbf{s}_k + \mbf{a}_w\rangle| &= |\langle \sum_{m \in \mathcal{S}}\mbf{a}_m,\mbf{s}_k\rangle + \langle \sum_{m \in \mathcal{S}}\mbf{a}_m,\mbf{a}_w \rangle|, \nonumber \\
    \mathllap{\Rightarrow\quad} &\leq |\alpha + j\zeta +  K\mu|, \nonumber \\
    \mathllap{\Rightarrow\quad} &\leq |\alpha +  j\zeta| + K\mu. 
\end{align}

MLMP correctly detects the column if,
\begin{align}
    |(\alpha+1) + j\zeta| - (K-1)\mu  &\geq |\alpha +  j\zeta| + K\mu, \nonumber \\
    \mathllap{\Rightarrow\quad} 1+2\alpha-(K-1)\mu &\geq K\mu, \nonumber \\
    \mathllap{\Rightarrow\quad} 1 - (K-1)\mu &\geq K\mu,  \nonumber\\
    \mathllap{\Rightarrow\quad} K \leq \frac{1+\mu}{2\mu}. \nonumber
\end{align}

\end{proof}

\bibliography{my_bib}
\bibliographystyle{ieeetr}
\newpage
\vfill

\end{document}